\documentclass[preprint,12pt]{elsarticle}

\usepackage{amssymb}
 \usepackage{amsmath}
\include{epsf}
\usepackage {epsfig}
\usepackage{graphicx}
\usepackage{setspace}
\newtheorem{theorem}{Theorem}[section]

\newtheorem{corollary}[theorem]{Corollary}

\newtheorem{proof}[theorem]{Proof}

\journal{Arxiv}

\begin{document}

\begin{frontmatter}


\title{On the Stability of Linear Discrete-Time Fuzzy Systems}

\author{Gabriele Oliva$^*$, Stefano Panzieri$^*$ and Roberto Setola$^{\dagger}$}

\address{$^*$ University Roma Tre of Rome, Italy. oliva@dia.uniroma3.it, panzieri@uniroma3.it\\
$^{\dagger}$ University Campus Bio-Medico of Rome, Italy. r.setola@unicampus.it}

\begin{abstract}
In this paper the linear and stationary Discrete-time systems with state variables and dynamic coefficients represented by fuzzy numbers are studied, providing some stability criteria, and characterizing the bounds of the set of solutions in the case of positive systems.
\end{abstract}


\end{frontmatter}


\section{Introduction}
\label{sec_introduction}

Transforming a real-world system into a deterministic model (e.g., a set of dynamic equations) often leads to errors whose effect may greatly influence the process and can not be neglected.

If the nature of errors is random or probabilistic, then it is possible to adopt a stochastic framework; however if the underlying structure is not probabilistic, for instance due to subjective modeling choices, a different formalism is required \cite{Laksh:2003}.

Fuzzy theory, first introduced in 1965 by L.A. Zadeh \cite{zadeh}, appears the most natural choice to model complex systems affected by vagueness.

In the literature many approaches have been introduced: in \cite{fn1,fn2} fuzzy numbers are used to model uncertain quantities (e.g., different beliefs or opinions \cite{fedrizzi}), considering also the fuzzy extension of traditional arithmetic operations; a different approach is to synthesize controllers based on fuzzy rules \cite{fc1,fc2,fc3,fc4} in order to control complex real-world systems encoding the experience of human operators (e.g., a parking system for a wheeled mobile robot \cite{fc5} or the navigation of a mobile robot \cite{ulivi}); a different, more formal approach is to consider a dynamic fuzzy equation or system, both in the continuous \cite{Trig,Laksh:2003,Pearson,Seikkala,Kay} and discrete-time \cite{Laksh:2003} fashion.

The role of Fuzzy theory appears of particular practical utility when humans are directed involved; for instance in \cite{Oliva}, systems with crisp dynamic and fuzzy state have been studied, and the framework has been applied to critical infrastructure related problems, where the only information available is that provided by operators and stakeholders and is therefore linguistic and vague.

In this paper the stability of discrete-time linear and stationary fuzzy systems where both the coefficients and the state are described by means of fuzzy variables is studied, allowing to represent uncertainty to a greater extent with respect to the approach in \cite{Oliva}.

In the continuous time case, the approach of representing the fuzzy system as a family of Differential Inclusions \cite{Cellina,Filippov,Laksh:2003} has been successfully adopted in \cite{Hull,Diamond,Lak0}; in this paper we will follow a similar path, introducing the framework of \emph{Fuzzy Difference Inclusions} (FDI).

The paper is organized as follows: after some preliminary definitions, Section 2 introduces the Discrete-time Fuzzy Systems and Section 3 specifies them as a family FDIs, while Section 4 address the stability of fuzzy systems described as a family of FDIs; Section 5 further specifies the results for linear systems, while Section 6 characterizes the solution set of positive linear systems; finally some conclusive remarks are collected in Section 7.

\subsection{Preliminaries}
\label{sec_preliminaries}
In the following, vectors and vectorial functions will be represented by boldface letters while scalars and functions with scalar codomain will be represented by plain letters.
Moreover, to avoid confusions, ${\bf x}$ will denote a vector with fuzzy entries, while {\em crisp} (i.e., non-fuzzy) vectors will be denoted by ${\bf z}$.

Let $I_p$ denote the $p\times p$ identity matrix and let ${\bf c}_p$ be a vector of $p$ components, each equal to $c$.

Let $\mathbb{R},\mathbb{N}$ be the set of reals and integers, respectively and $\mathbb{R}^+,\mathbb{N}^+$ be the set of nonnegative real and integer numbers, respectively.
Let $\mathbb{K}_C^N$ be the space of nonempty compact convex subsets of $\mathbb{R}^N$.
Let $\mathbb{B}^N$ be the open unit ball in $\mathbb{R}^N$.

Given a space $\mathbb{X}$ and a particular distance $d$ defined over such a space, the Hausdorff separation and the Hausdorff metric for two sets $A,B \subset \mathbb{X}$ are given by: 
\begin{eqnarray}
\rho^*_{d,\mathbb{X}}(A,B)=sup\{d(a,b) : a\in A, b \in B\} \\
\rho_{d,\mathbb{X}}(A,B)=\max \{\rho^*_{d,\mathbb{X}}(A,B), \rho^*_{d,\mathbb{X}}(B,A) \}
\end{eqnarray}

In this work it will be considered the following distance in $\mathbb{R}^{N}$:
\begin{equation}
\label{distrn}
d_{\mathbb{R}^N}({\bf z_1},{\bf z_2})=\sum_{i=1}^N d_{\mathbb{R}}(z_{1i},z_{2i})
\end{equation}
where $d_{\mathbb{R}}(\cdot, \cdot)$ is the euclidean distance in $\mathbb{R}$ and ${\bf z_1},{\bf z_2}\in \mathbb{R}^N$.

If a matrix M is non-negative (non-positive), i.e., it has only non-negative (non-positive) entries, write $M\geq 0$ ($M\leq 0$).
If $B-A$ is a nonnegative matrix, write $B\geq A$.

\section{Fuzzy Discrete-time Systems and Inclusions}
The key idea of fuzzy theory is to extend traditional set theory by allowing an element to partially belong to a set.
Therefore, a fuzzy subset of $\mathbb{R}$ is defined by a \emph{membership function} $x:\mathbb{R} \rightarrow [0,1]$ which assigns to each point $p \in \mathbb{R}$ a grade of membership in the fuzzy set \cite{zadeh}.
Notice that, when dealing with fuzzy sets, it is very common to denote with the same symbol the fuzzy sets itself and the membership function that defines the set, because the membership function univocally determines the set and vice-versa.

Let the $p$-membership $x(p)$ be defined as the grade of membership of $p\in \mathbb{R}$ in the set $x$. In the following, the value assumed by a given fuzzy variable at time instant $k$ will be denoted by $x(k)$; in this case the $p$-membership of the set $x(k)$ is denoted by $x(p,k)$.

For each $\alpha\in(0,1]$, the {\em $\alpha$-level} set $[x]^\alpha$ of a fuzzy set is the subset of points $p \in\mathbb{R}$ with membership grade $x(p)\geq \alpha$.
The support $[x]^0$ of a fuzzy set is defined as the closure of the union of all its $\alpha$-level sets (see Figure \ref{fig:alfa}).

\begin{figure}[!ht]
\begin{center}
\includegraphics[scale=0.55]{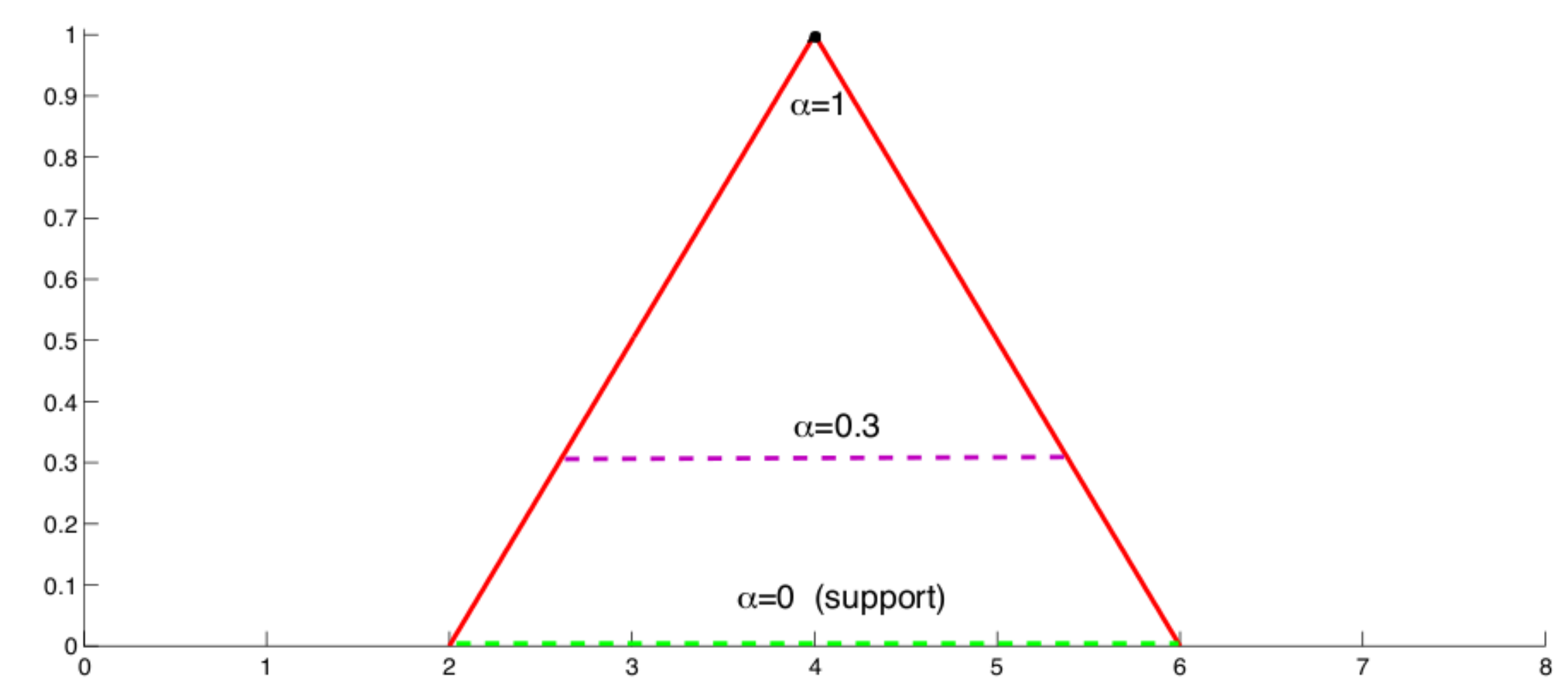}
\caption{Different $\alpha$-levels of a triangular-shaped fuzzy membership function. }
 \label{fig:alfa}
\end{center}
\end{figure}

The use of $\alpha$-levels, as shown in Figure \ref{fig:alfa}, allows to treat fuzzy sets as a set of nested real intervals; in the following the evolution of a fuzzy dynamic system will be evaluated {\em level-wise} by considering for each $\alpha$, the evolution of a system whose state is described by an interval (i.e., the corresponding $\alpha$-level).
In \cite{Laksh:2003} $\mathbb{E}$ is introduced as the space of all fuzzy subsets $x$ of $\mathbb{R}$ such that:
\begin{enumerate}
\item $x$ maps $\mathbb{R}$ onto $[0,1]$;
\item $[x]^0$ is a bounded subset of $\mathbb{R}$;
\item $[x]^\alpha$ is a compact subset of $\mathbb{R}$ for all $\alpha \in (0,1]$;
\item $x$ is \emph{fuzzy convex}, that is: $x(\phi p + (1-\phi)q) \geq min [x(p),x(q)]$ for all $p,q \in \mathbb{R}$
\end{enumerate}

These condition are equivalent to require that the membership functions are compact, that their shape is composed of a monotone nondecreasing and a monotone non increasing part (e.g., a triangle), and the $\alpha$-levels are nested, i.e., those with greater values of $\alpha$ are contained into those with smaller $\alpha$ (all the intervals are contained in the support).
With this formulation the fuzzy sets in $\mathbb{E}$ are often called \emph{Fuzzy Numbers} (FN) \cite{fn1,fn2}.
%

A \emph{triangular} fuzzy number (TFN) $x\in \mathbb{E}$, in particular, is described by an ordered triple $\{x_l,x_c,x_r\} \in\mathbb{R}^3$ with $x_l\leq x_c\leq x_r$ and such that $[x]^0=[x_l,x_r]$ and $[x]^1=\{x_c\}$, while in general the $\alpha$-level set is given, for any $\alpha \in [0,1]$ by:
\begin{equation}
[x]^\alpha = [x_c - (1-\alpha)(x_c -x_l), x_c + (1-\alpha)(x_r-x_c)]
\end{equation}

An illustrative example of triangular fuzzy numbers is reported in Figure \ref{fig:incertezza}; the figure shows two triangular fuzzy numbers with different levels of ambiguity, representing for example the codification of the statement ``an high wall". Specifically, the blue internal triangle represents the fuzzy number corresponding to the measure of $4$ meters with a given degree of ambiguity (i.e., the width of the base of the triangle);  the red external triangle still represents the same fuzzy number $4$, but with a greater ambiguity.

The triangular representation is not the sole available alternative; as depicted in Figure \ref{fig:shape} many other shapes are possible, and more complex is the shape, more descriptive is the resulting fuzzy number (i.e., the vagueness is  better characterized). For instance the existence of a plateau for a given interval represents complete indeterminacy for that interval, or an asymmetry with respect to the peak may represents different beliefs in the entity of the values that are smaller or bigger than the ``nominal" value (e.g., the best and worst cases).
However, in many applications, the Triangular Fuzzy Numbers are the most used, because they can be described by the triple of the abscissae of their vertices ($\{2,4,6\}$ and $\{3,4,5\}$ in the case of the Triangular FN in Figure \ref{fig:incertezza}); moreover, as shown in Figure \ref{fig:alfa}, a crisp value is obtained for $\alpha=1$, hence the level with the strongest belief collapses into a point. 

More complex shapes, as shown in Figure \ref{fig:shape}, allow to better characterize the uncertainty: the leftmost Fuzzy Number represents a situation where uncertainty rapidly decreases while approaching the peak value;  the rightmost Fuzzy Number, due to its trapezoidal shape, models the case where a single value with maximum belief can not be found.
Notice further that, as stated before, the shape of a FN needs not to be symmetric (see the central TFN in Figure \ref{fig:shape}), thus allowing to represents different beliefs on the left and right spread of uncertainty with respect to the value associated with the maximum belief.

\begin{figure}[!ht]
\begin{center}
\includegraphics[width=5in]{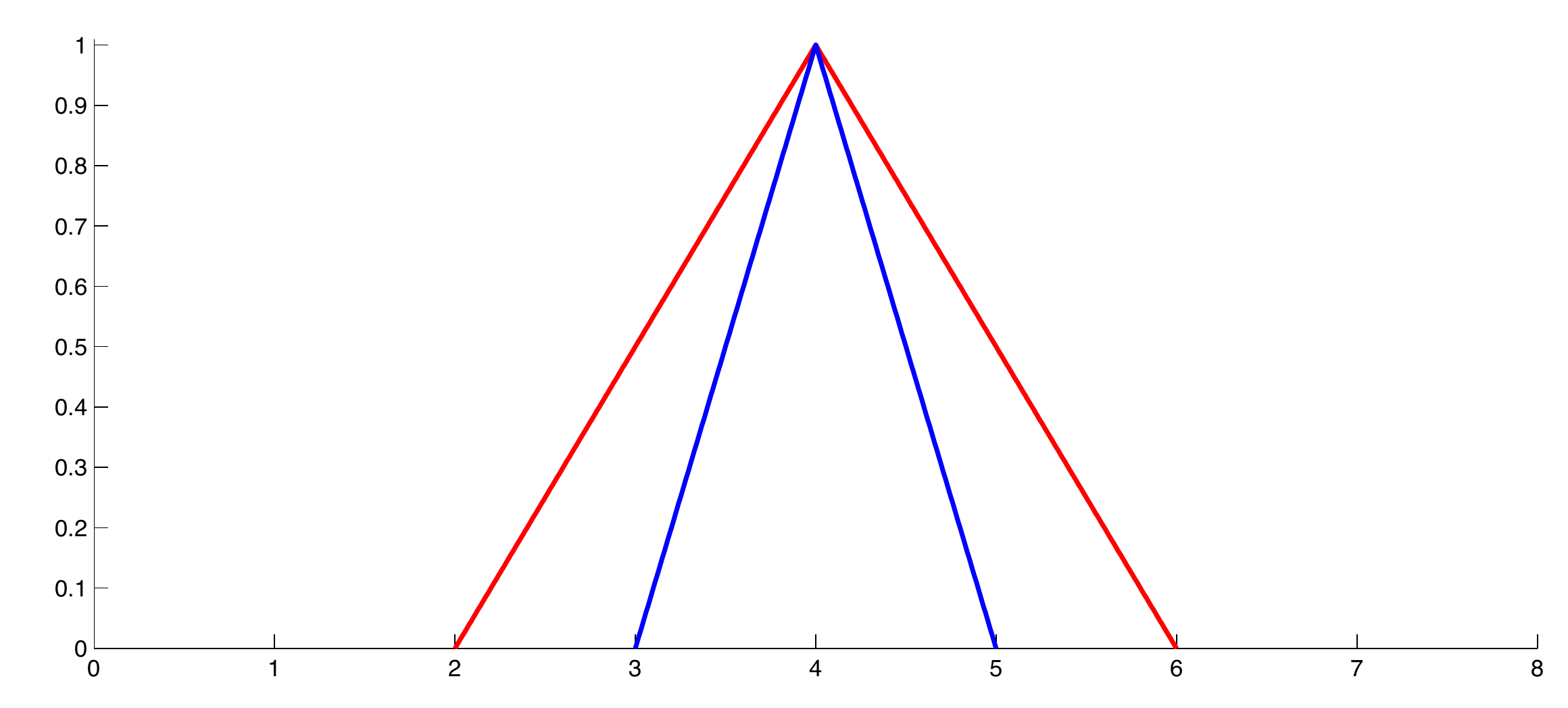}
\caption{ Two triangular fuzzy numbers representing the same value ``about 4", although with different uncertainty.}
 \label{fig:incertezza}
\end{center}
\end{figure}

\begin{figure}[!ht]
\begin{center}
\includegraphics[width=5in]{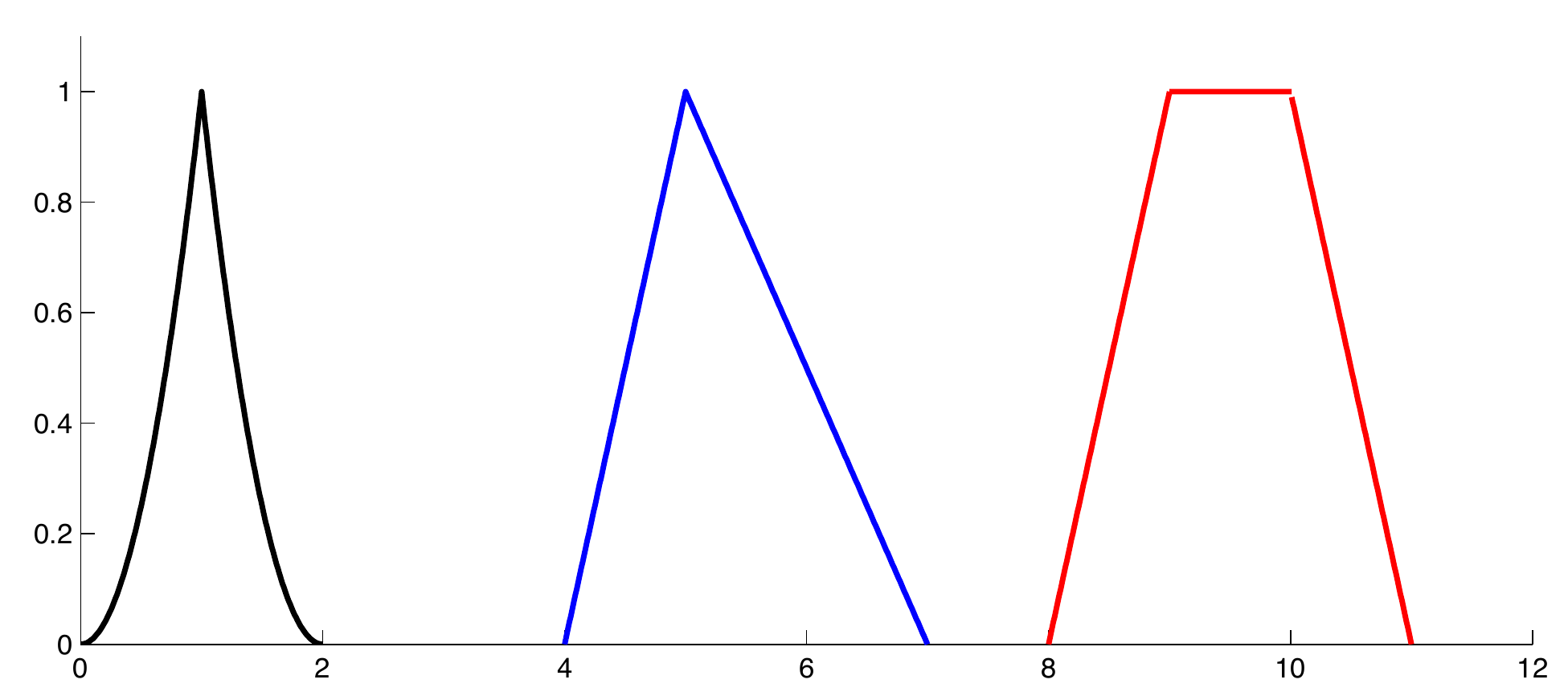}
\caption{Examples of fuzzy numbers in the space $\mathbb{E}$.}
 \label{fig:shape}
\end{center}
\end{figure}

In order to measure the ``distance" between two fuzzy sets, the space $\mathbb{E}$ is equipped with the metric \cite{Laksh:2003}:
\begin{equation}
d_\mathbb{E}(x_1,x_2)=\sup \{|x_1(p)-x_2(p)| : p \in \mathbb{R}\}
\end{equation}
which measures the largest difference in the membership grades of two fuzzy sets $x_1$ and $x_2$; clearly $d_\mathbb{E}(x_1,x_2)\in [0,1]$.
Figures \ref{fig:overlapping1} and \ref{fig:overlapping2} provide examples of computation of such a distance; it is worth to notice that, unless the maximum $\alpha$-level (i.e., $a=1$) of one of the two FNs coincides with an interval where at least a point in the other FN is non-zero, then  $d_\mathbb{E}(x_1,x_2)=1$ (see Figure  \ref{fig:overlapping1}); when such a condition is verified, the metric assumes a value that is inversely proportional to the degree of overlapping of the two FNs (see Figure \ref{fig:overlapping2}); finally, notice that the distance becomes zero if and only if the shape of the two FNs coincide.

\begin{figure}[!ht]
\begin{center}
\includegraphics[width=5in]{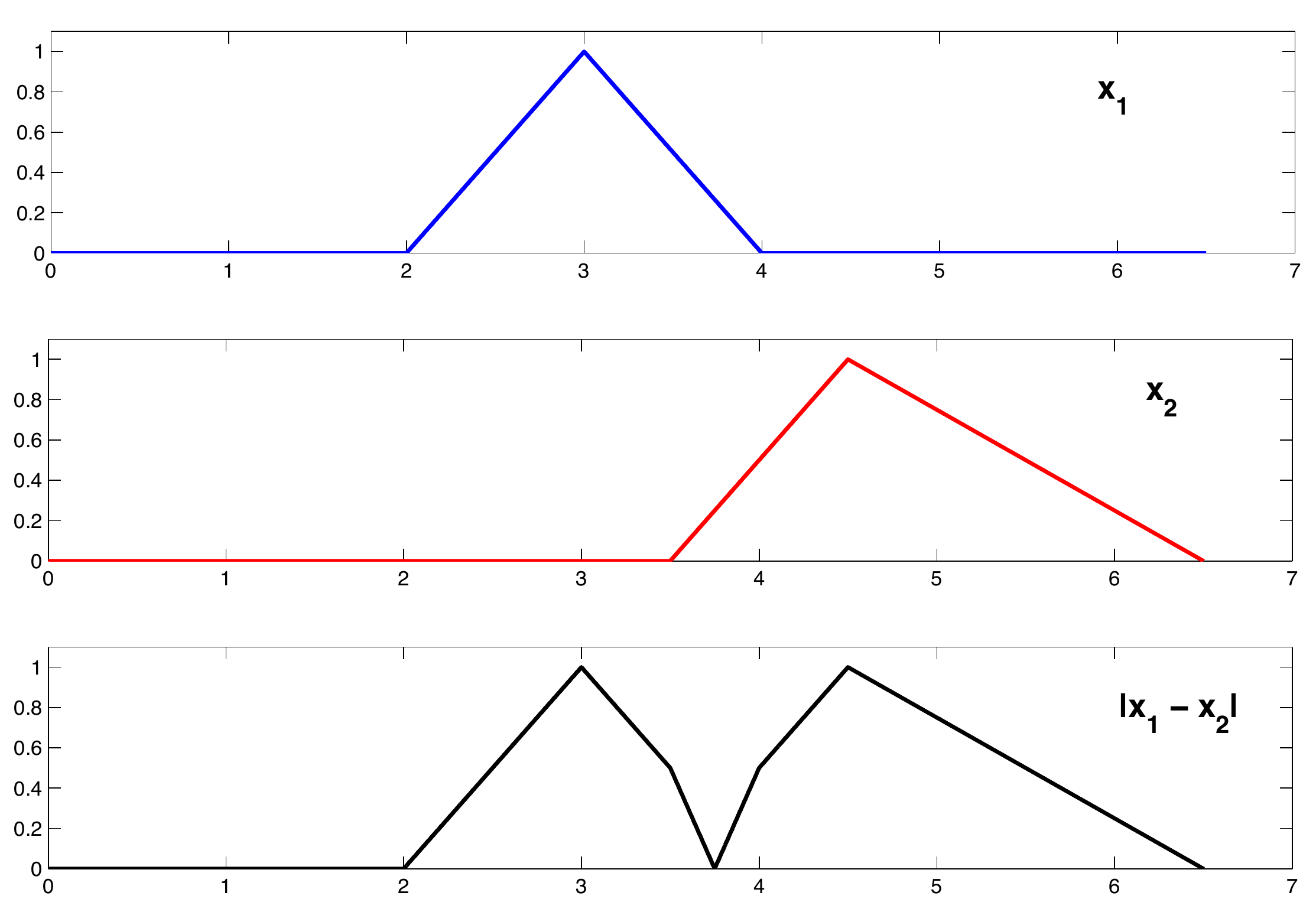}
\caption{Example of distance calculation for two TFNs $x_1=\{2,3,4\}$ and $x_2=\{3.5,4.5,6.5\}$. The distance $d_{\mathbb{E}}(x_1,x_2)$ is equal to $1$.}
 \label{fig:overlapping1}
\end{center}
\end{figure}

\begin{figure}[!ht]
\begin{center}
\includegraphics[width=5in]{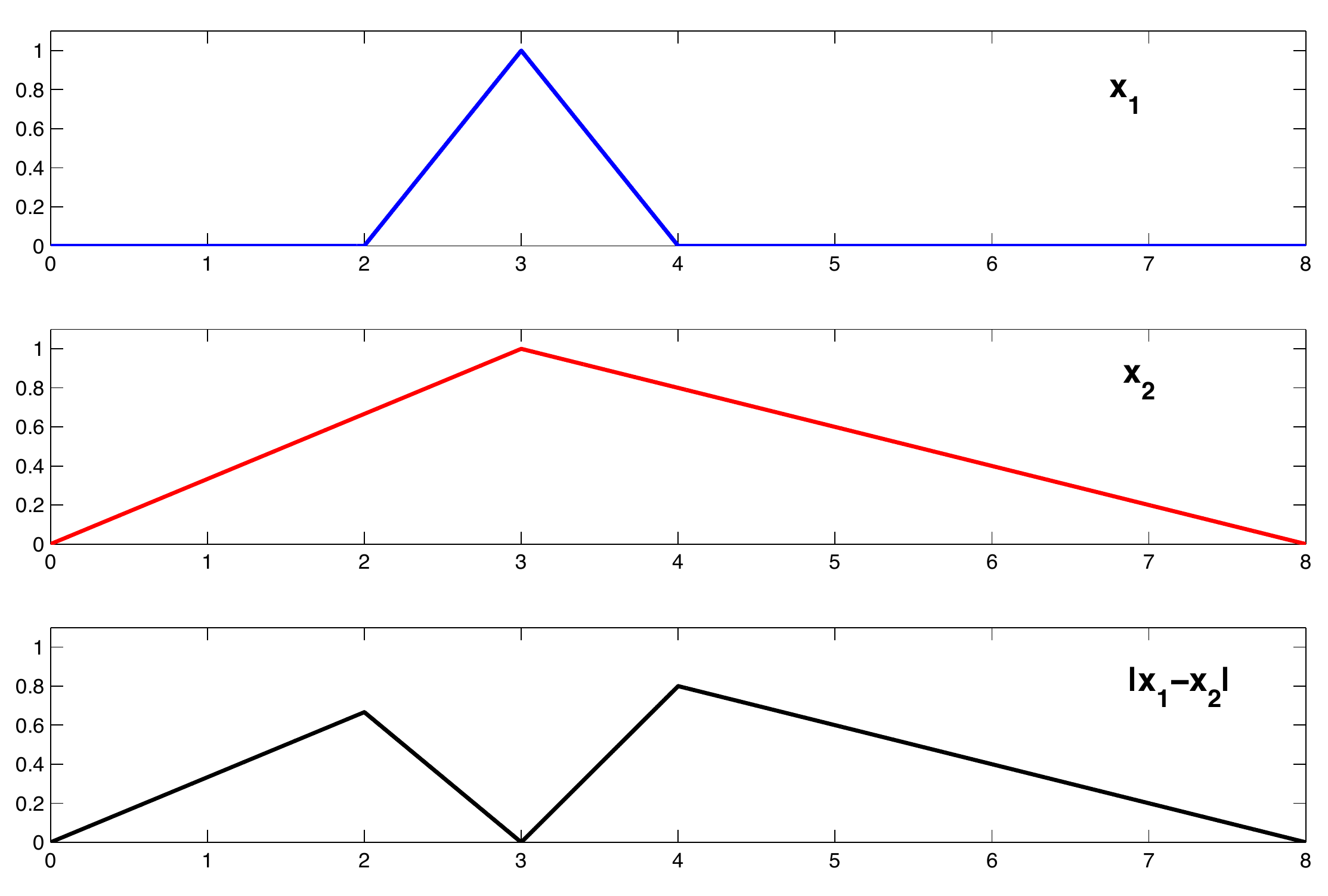}
\caption{Example of distance calculation for two TFNs $x_1=\{2,3,4\}$ and $x_2=\{0,3,8\}$. The distance $d_{\mathbb{E}}(x_1,x_2)$ is equal to $0.8$.}
 \label{fig:overlapping2}
\end{center}
\end{figure}

The following equation correlates the distance with the $\alpha$-sets \cite{Laksh:2003}, considering the euclidean distance d in $\mathbb{R}$:
\begin{eqnarray}
\label{hausdist}
d_\mathbb{E}(x_1,x_2)=\sup_{\alpha >0} \{\rho_{d,\mathbb{R}}([x_1]^\alpha,[x_2]^\alpha)\} & x_1,x_2 \in \mathbb{E}
\end{eqnarray}

In the following the concept of {\em level-wise convergence} of a sequence of fuzzy numbers will be defined, that is, for a fixed $\alpha \in [0,1]$, the convergence of the $\alpha$-levels of the FNs in the sequence.

Let $\{x_n\}$ be a sequence on $\mathbb{E}$, then $\{x_n\}$ converges level-wise to $x \in \mathbb{E}$ if, for all $\alpha \in (0,1]$:
\begin{eqnarray}
\rho_{d_\mathbb{R},\mathbb{R}}([x_n]^\alpha,[x]^\alpha)\rightarrow 0 & \mbox{ as } & n\rightarrow \infty
\end{eqnarray}
Define  
\begin{equation}
\label{eq:psieq0}
\Psi =  \{ x \in \mathbb{E}: x(\phi p + (1-\phi) q) \geq \phi x(p) + (1-\phi)x(q)\}
\end{equation}
for each $p, q \in [x]^0, \phi \in [0,1]$.

In order to consider vectors of $N$ components, each being a FN, the space $\mathbb{E}^{N}$ has to be characterized; to this end $\mathbb{E}^{N}$ is equipped with the following metric:
\begin{equation}
\label{dizero}
d_{\mathbb{E}^N}({\bf x},{\bf y}) = \sum_{i=1}^N d_\mathbb{E}(x_i,y_i)
\end{equation}
where ${\bf x}=[x_1, \ldots,x_n]^T$ and ${\bf y}=[y_1, \ldots,y_n]^T$, ${\bf x},{\bf y}\in \mathbb{E}^n$, while $\forall i=1, \ldots,N$ $x_i,y_1 \in  \mathbb{E}$.

Define the $\alpha$-level of a vector of FNs ${\bf x}\in \mathbb{E}^N$ as the set of vectors ${\bf z}\in \mathbb{R}^N$ such that, $\forall i =1,\ldots, N$ $z_i $ belongs to the $\alpha$-level of $i$-th component $x_i$.

It is worth to notice that many applications based on fuzzy theory often handle uncertain models implicitly, by considering approximations of the algebraic operations extended to fuzzy numbers; however these approaches are often limited to triangular shapes, due to computational constraints and due to the complexity of the approximations if the shape is not triangular. 
Other approaches are typically aimed to synthesize controllers based on fuzzy rules \cite{fc1,fc2,fc3,fc4} in order to control complex real-world systems encoding the experience of human operators (e.g., a parking system for a wheeled mobile robot \cite{fc5} or the navigation of a mobile robot \cite{ulivi}).
Within all these approaches, however, it is not possible (or very hard) to obtain a closed form, neither in the simplest cases, and it is therefore very difficult to study the stability of the fuzzy systems.

In order to clarify the limits of simulation-based techniques, consider for instance the scalar system 
\begin{equation}x(k+1)=mx(k)+q\end{equation}
where $x(0)=\{x_{0l},x_{0c},x_{0r}\}$, $c=\{c_{l},c_{c},c_{r}\}$ and $q=\{q_{l},q_{c},q_{r}\}$ are triangular fuzzy numbers.
In the simulative perspective, the extension of the sum to TFNs is immediate;  indeed, such an operation is linear and the resulting fuzzy number is also triangular \cite{fn1}. 
The sum of two TFNs $=[a_l, a_c, a_r]$ and $b=[b_l, b_c, b_r]$ is given by
\begin{equation}
\label{fuzzysum}
a+b=[a_l + b_l, a_c + b_c, a_r + b_r]
\end{equation}

Analogously, the multiplication by a real scalar $\beta$ is given by
\begin{equation}
\label{fuzzyscalarmult}
\beta a=[\beta a_1, \beta a_2, \beta a_3]
\end{equation}

Triangular fuzzy numbers, therefore, are closed with respect to sum and product with a real scalar \cite{fn1}.

When addressing the product of two TFNs, however, there is the need to introduce approximations, since the product of two TFNs is in general not triangular;
in the literature many approximations of the product have been given, and the most simple,  in the case where $a,b \geq 0$ (i.e., $a_l,b_l \geq 0$), is given by:
\begin{equation}
\label{fuzzyprod}
a*b=[a_l*b_l, a_c*b_c, a_r*b_r]
\end{equation}
However, in this way, besides loosing expressivity (because only triangular numbers can be typically handled by simulators), it is not possible to analyze and formally characterize the stability of the system.

Moreover the approximated methods fail to address more complex situations, where, besides the products, a general nonlinear fuzzy-valued function with fuzzy argument has to be computed.

In the following  we will consider {\it Complete Discrete-Time Fuzzy System} (C-DFS), defined as follows:
\begin{equation}
\label{fuzzysystemfull}
{\bf x}(k+1) = {\bf H}({\bf x}(k),k), \quad {\bf x}(0)={\bf x}_0
\end{equation}
where ${\bf H}:\mathbb{N}^+ \times \mathbb{E}^N \rightarrow \mathbb{E}^N$ and ${\bf x},{\bf x}_0 \in \mathbb{E}^N$.
Notice that for System (\ref{fuzzysystemfull}), even if the initial conditions were crisp, due to the fuzziness of the dynamic the state of the system is composed of fuzzy variables.

However, as noted before, the dynamic of system (\ref{fuzzysystemfull}) generates evolutions that are in general complex and not triangular, and the level-wise evaluation of such a system is particularly hard, as will be explained in the following; this makes unfeasible any direct calculation of (\ref{fuzzysystemfull}).

An elegant solution can be achieved in the framework of Difference Inclusions \cite{Kellett0,Kellett}.

\section{Fuzzy Difference Inclusions}
In the continuous time case, the approach of representing a fuzzy system as a {\em Fuzzy Differential Inclusion} (FDI) that is, a family of difference inclusions \cite{Cellina,Filippov} (one for each $\alpha$-level), has been successfully adopted in \cite{Hull,Diamond,Lak0}.

Defining $[{\bf H}({\bf x}(k),k)]^\alpha={\bf Q}([{\bf x}(k)]^\alpha,k,\alpha)$, the C-DFS (\ref{fuzzysystemfull}) can be rewritten  as the FDI:
\begin{equation}
\label{incl_fuzzy1}
{\bf x}^{\alpha}(k+1)\in {\bf Q}([{\bf x}(k)]^\alpha,k,\alpha); \quad {\bf x}^{\alpha}(0)\in[{\bf x}_0]^\alpha; \quad   \alpha\in J 
\end{equation}
where $J=[0,1]$, $\Omega$ is an open subset of $\mathbb{R}^N \times \mathbb{N}^+$ containing $[{\bf x}_0]^\alpha \times \{0\}$ and ${\bf Q}:\Omega \times J \rightarrow \mathbb{K}_C^N$ is a continuous and set valued map.
This latter condition is equivalent to require that the set-valued function ${\bf Q}([{\bf x}(k)]^\alpha,k,\alpha)$ is defined for a neighborhood of the initial conditions, and its image is compact and convex.
Notice that in (\ref{incl_fuzzy1}) ${\bf x}^{\alpha}(k+1),{\bf x}^{\alpha}(0) \in \mathbb{R}^N$ and hence $[{\bf x}(k+1)]^\alpha$ is the set obtained considering all the attainable ${\bf x}^{\alpha}(k+1)$, starting from all the ${\bf x}^{\alpha}(k) \in [{\bf x}(k)]^\alpha$.

From now on, it is assumed that all maps are {\em proper}, that is have nonempty images of points in their domain. 

The idea of considering difference inclusions becomes clearer while considering a simple example; consider the system 
\begin{equation}
x(k+1)=mx(k) \quad x(0)=x_0
\end{equation}
where $m=\{m_l,m_c,m_r\}$ and $x_0=\{x_{0l},x_{0c},x_{0r}\}$ are fuzzy.
As stated before, the product does not yield to a triangular fuzzy number and is very hard to evaluate; this is even more true in a level-wise representation.
Indeed, for each $\alpha$-level we have to take into account the product of two intervals, $[\underline m^\alpha, \bar m^\alpha]$ and $[\underline x^\alpha(k), \bar x^\alpha(k)]$ and the approach used for simple DFS system cannot be adopted. 
A solution is to consider, for each $\alpha$-level and for each point belonging to the interval $[\underline m^\alpha, \bar m^\alpha]$, a level-wise representation of a standard DFS, thus obtaining the difference inclusion
\begin{equation}
x^\alpha(k+1) \in [m]^\alpha[x(k)]^\alpha
\end{equation}

In order to proceed there is the need to recall some basic notions on difference inclusions and their stability.

\subsection{Difference Inclusions}
A Difference Inclusion \cite{Kellett0,Kellett} is given by
\begin{equation}
\label{eq:di}
{\bf z}(k+1) \in {\bf P}({\bf z}(k),k), \quad {\bf z}(0)\in {\bf z}_0
\end{equation}

\noindent where ${\bf z}_0\subset \mathbb{R}^N$, $\Omega$ is an open subset of $\mathbb{R}^N \times \mathbb{N}^+$ containing ${\bf z}(0)\times \{0\}$  and ${\bf P}:\Omega \rightarrow \mathbb{K}_C^N$ is a continuous set valued map.
Notice that the dynamic ${\bf P}$ is set valued and also the initial condition is a set.

Notice further that, in a very general perspective, difference and differential inclusions are often required to be {\em upper} or {\em lower semicontinuous}, since it is possible to study the behavior of inclusions with discontinuous right-end side; however such an extension is out of the scope of this work, and the reader may refer to \cite{Laksh:2003,Diamond,Cellina,Filippov} for a more general approach.

Let $\mathbb{Z}_{k}(\mathbb{R}^N)$ be the set of continuous functions $f:\mathbb{R}^N \times \{0, \ldots, k\}\rightarrow \mathbb{R}^N$. We need to define the following sets:
\begin{itemize}
\item {\bf Set of Solutions}: the \emph{set of solutions} $\mathbb{S}({\bf z}_0, k)$ of the inclusion (\ref{eq:di}) on $\{0,\ldots,k\}$, is the set of all solutions of the inclusion (\ref{eq:di}) from time step $0$ until $k$; clearly $\mathbb{S}({\bf z}_0, k)\in \mathbb{Z}_{k}(\mathbb{R}^N)$.
\item {\bf Attainable Set}: the \emph{attainable set}, for a given step $k$, is the set of values that the solutions of the inclusion (\ref{eq:di}) assume in $k$ and is defined as
\begin{equation}
\mathbb{A}({\bf z}_0, k)=\{{\bf z}(k): {\bf z}(\cdot) \in \mathbb{S}({\bf z}_0, k)\}
\end{equation} 
\end{itemize}

%

A set $\mathbb{M}$ is {\em stable} for the Inclusion (\ref{eq:di}) if for all $\epsilon >0$ there exists $\delta=\delta(\epsilon)>0$ such that ${\bf z}_0 \in \mathbb{M}+\delta \mathbb{B}^N$ implies ${\bf z}(k)\in \mathbb{M}+\epsilon \mathbb{B}^N$ for all $k\in \mathbb{N}^+$ and for every solution in $\mathbb{S}({\bf z}_0, k)$. 

The above definition may be rephrased as ${\bf z}_0\in \mathbb{M}+\delta \mathbb{B}^N$ implies that 
$\rho_{d_{\mathbb{R}^N},\mathbb{R}^N}(\mathbb{A}({\bf z}_0,k),\mathbb{M})\leq \epsilon$  for all $k\in \mathbb{N}^+$, where $\rho_{d_{\mathbb{R}^N},{\mathbb{R}^N}}(\cdot,\cdot)$ is the Hausdorff distance in $\mathbb{R}^N$. 

If $\rho_{d_{\mathbb{R}^N},\mathbb{R}^N}(\mathbb{A}({\bf z}_0,k),M)\rightarrow 0$ as $k\rightarrow \infty$ and M is  stable then it is asymptotically stable.

\section{Stability of Fuzzy Difference Inclusions}
\label{sec_fuzzy1}

Let us now discuss the characteristics of a C-DFS (\ref{fuzzysystemfull}) expressed as a FDI (\ref{incl_fuzzy1}).

Denote the set of solutions of an inclusion in the family (\ref{incl_fuzzy1}) (i.e., for a given $\alpha\in [0,1]$) by 
$\mathbb{S}_\alpha({\bf x}_0,k)$ and the attainable set by $\mathbb{A}_\alpha ({\bf x}_0, k)$, while $\mathbb{S}^f({\bf x}_0,k)$ and $\mathbb{A}^f({\bf x}_0, k)$ are the set of solutions and the attainable set for the whole family (\ref{incl_fuzzy1}).

Note that, during the evolution of a FDI, it is not in general verified that the convexity of the $\alpha$-levels is preserved; 
therefore the stability of FDIs is often studied in  the space $\mathbb{D}^N$ of fuzzy sets with not necessarily convex level sets \cite{Laksh:2003,Hong};however in this work we relax the semicontinuity assumption, supposing that the dynamic is continuous, hence avoiding the issue of non-convex level sets \cite{Hong}.

The following {\em Stacking Theorem} \cite{Laksh:2003}, whose statement is reported below, characterizes the elements of $\mathbb{D}$, but can be easily generalized to the vectorial case or to any Banach space.
\begin{theorem}
\label{stacking}
Let $\{Y_\beta \subset \mathbb{R}: 0\leq \beta \leq 1\}$ be a family of compact subsets satisfying
\begin{enumerate}
\item $Y_\beta \subseteq Y_\alpha$ for $0 \leq \alpha \leq \beta \leq 1$
\item $Y_\beta = \cap_{i=1}^\infty Y_{\beta_i}$ for any nondecreasing sequence $\beta_i \rightarrow \beta$ in $[0,1]$
\end{enumerate}
Then there is a fuzzy set $u\in \mathbb{D}$ such that $[u]^\alpha=Y_\alpha$. If the $Y_\beta$ are also convex, then $u \in \mathbb{E}$. Conversely the level set of any $u\in \mathbb{E}$ are compact and convex and satisfy these conditions.
\end{theorem}

The above theorem states that a family of compact and closed nested intervals coincides with a fuzzy set in $\mathbb{D}$, and the property can be trivially extended to the vectorial case (if the property is true for each state variable); moreover a fuzzy set in $\mathbb{E}^N$ is always composed of level sets that are non increasing in $\alpha$.

There is the need to introduce the following stability definitions, that extend to the fuzzy domain the definitions provided for difference inclusions. Notice that, although ${\bf x}_0$ is a point in $\mathbb{E}^N$, the stability definitions involve the fuzzy attainable set, therefore the stability is defined for a set $\mathbb{U}\subset \mathbb{E}^N$.

A set $\mathbb{U} \subset \mathbb{E}^N$ is \emph{stable} for System (\ref{incl_fuzzy1}) if for all $\epsilon >0$ there exists a $\delta(\epsilon)$ such that 
\begin{equation}
{\bf x}_0\in \mathbb{U}+\delta \mathbb{B}_E^N \rightarrow \mathbb{A}^f({\bf x}_0,k)\in \mathbb{U} + \epsilon \mathbb{B}_E^N
\end{equation} 
for all $k\in \mathbb{N}^+$, where $\mathbb{B}_{E}^N$ is the open unit ball in $\mathbb{E}^N$.

In other terms 
\begin{equation}
\rho^*_{d_{\mathbb{E}^N},\mathbb{E}^N}({\bf x}_0,\mathbb{U})<\delta \Rightarrow \rho_{d_{\mathbb{E}^N},\mathbb{E}^N}(\mathbb{A}^f({\bf x}_0,k),\mathbb{U})\leq \epsilon
\end{equation}
for all $k\in \mathbb{N}^+$, where $\rho^*_{d_{\mathbb{E}^N},\mathbb{E}^N}(\cdot,\cdot)$ and $\rho_{d_{\mathbb{E}^N},\mathbb{E}^N}(\cdot,\cdot)$ are the Hausdorff distance and the Hausdorff separation in $\mathbb{E}^N$, respectively \cite{Laksh:2003}, based on a distance $d_{\mathbb{E}^N}$, defined in (\ref{dizero}).

If $\rho_{d_{\mathbb{E}^N},\mathbb{E}^N}(\mathbb{A}^f({\bf x}_0,k),\mathbb{U})\rightarrow 0$ as $k \rightarrow \infty$ and $\mathbb{U}$ is stable, then it is  asymptotically stable.

The following theorem characterizes the structure of the set of solutions and the attainable set of a FDI (\ref{incl_fuzzy1}).

\begin{theorem}
\label{attsol}
Let $\mathbb{U}\subset \mathbb{D}^N$ and suppose that $[\mathbb{U}]^0 \subset \mathbb{R}^N$ is bounded.
The attainable sets $\mathbb{A}_\alpha({\bf x}_0,k)$ of the family of inclusions (\ref{incl_fuzzy1}) are the level sets of a fuzzy set $\mathbb{A}^f({\bf x}_0,k)$ and the sets of solutions $\mathbb{S}_\alpha({\bf x}_0,k)$ are the level sets of a fuzzy set $\mathbb{S}^f({\bf x}_0,k)$. 
\end{theorem}

\begin{proof}
The set $\mathbb{S}_\alpha({\bf x}_0,k)$ is composed of solutions ${\bf x}^\alpha(\cdot)$ such that ${\bf x}^\alpha(k') \in \mathbb{A}_\alpha({\bf x}_0,k')$ for each $k'\in \{0, \ldots, k\}$.

To show that the sets $\mathbb{S}_\alpha({\bf x}_0,k)$ for each $\alpha \in [0,1]$ are compact we need to show that $\mathbb{A}_\alpha({\bf x}_0,k')$ is compact for each $k'\in \{0, \ldots, k\}$.
It is a well known result that compactness is invariant to continuous transformations (see \cite{hockingtopology} for instance); since $[{\bf x}(0)]^\alpha$ is compact, it follows that $\mathbb{A}_\alpha({\bf x}_0,1)$ is compact. Iterating it is verified that all the $\mathbb{A}_\alpha({\bf x}_0,k')$ are compact and hence $\mathbb{S}_\alpha({\bf x}_0,k)$ is compact.

Notice that ${\bf Q}([{\bf x}(k)]^\alpha,k,\alpha) \subseteq {\bf Q}([{\bf x}(k)]^0,k,0)$, $[{\bf x}]^\alpha \subseteq [{\bf x}]^0$ for all $\alpha \in [0,1]$, therefore the sets $\mathbb{S}_\alpha({\bf x}_0,k)$ are nonincreasing in $\alpha$.

Let $\{\beta_i\}$ be a nondecreasing sequence in $[0,1]$ converging to $\beta \in [0,1]$ and consider the sequence $\{ \mathbb{S}_{\beta_i}({\bf x}_0,k)\}$; clearly the intersection of the first $j$ elements in the sequence coincides with $\mathbb{S}_{\beta_j}({\bf x}_0,k)$.
Therefore the intersection of all the elements in the sequence gives
\begin{equation}
\bigcup_{i} \mathbb{S}_{\beta_i}({\bf x}_0,k) = \mathbb{S}_{\beta}({\bf x}_0,k)
\end{equation}
where, for infinite sequences, the equivalence holds in the limit of the infinite intersection.
A similar argument proves the result for the attainable sets  $\mathbb{A}_\alpha({\bf x}_0,k)$.

It is now possible to apply the Stacking Theorem (\ref{stacking}) on the space $\mathbb{Z}_k(\mathbb{R}^N)$ for $\mathbb{S}_{\beta}({\bf x}_0,k)$ and on the space $\mathbb{D}^N$ for $\mathbb{A}_\alpha({\bf x}_0,k)$, proving the statement.

\end{proof}

The following theorem characterizes the stability of a FDI (\ref{incl_fuzzy1}).
\begin{theorem}
\label{stabfuz}
Let $\mathbb{U}\subset \mathbb{D}^N$ and suppose that $[\mathbb{U}]^0 \subset \mathbb{R}^N$ is bounded  and 
the inclusion
\begin{equation}
\label{aaaa_di}
[{\bf x}]^{0}(k+1)\in {\bf Q}([{\bf x}]^0(k),k,0); \quad [{\bf x}]^{0}(0)=[{\bf x}_0]^0
\end{equation}
is (asymptotically stable); then the FDI (\ref{incl_fuzzy1}) is (asymptotically) stable.
\end{theorem}

\begin{proof}
From Theorem \ref{attsol}, the FDI has a fuzzy attainable set $\mathbb{A}^f({\bf x}_0,k)$ which coincides with $[\mathbb{A}^f({\bf x}_0,k)]^0$, because the other level sets are nested in the support.

Since the difference inclusion (\ref{aaaa_di}) is stable (there exists a stable set $\mathbb{M}$), then  $\mathbb{A}^f({\bf x}_0,k)$ is bounded and the family (\ref{incl_fuzzy1}) is stable. 

If the inclusion (\ref{aaaa_di}) is asymptotically stable, then  
\begin{equation}
\rho_{d_{\mathbb{R}^N},\mathbb{R}^N}(\mathbb{A}_0({\bf x}_0,k),M)\rightarrow 0 \mbox{ as } k\rightarrow \infty
\end{equation}

Therefore, since $[\mathbb{A}^f({\bf x}_0,k)]^\alpha \subseteq \mathbb{A}^f({\bf x}_0,k)$ for $0\leq \alpha \leq 1$ it follows that
 \begin{equation}
 \rho_{d,\mathbb{R}^N}([\mathbb{A}^f({\bf x}_0,k)]^\alpha,M)\rightarrow 0 \mbox{ as }k\rightarrow \infty
 \end{equation}
 and the fuzzy system (\ref{incl_fuzzy1}) is asymptotically stable.
\end{proof}

The above theorem proves that the stability of the level-wise representation of the support of the fuzzy system (i.e., $\alpha=0$) is a sufficient condition for the stability of the whole fuzzy system.

We can now provide some further characterization in  the case of linear and stationary FDIs.

\section{Linear and Stationary Fuzzy Difference Inclusions}
In this section we will consider Linear and Stationary C-DFS, defined as follows:
\begin{equation}
\label{linfuzzysystemfull}
{\bf x}(k+1) = {\bf H}{\bf x}(k), \quad {\bf x}(0)={\bf x}_0
\end{equation}
where ${\bf H}$ is a $n\times n$ fuzzy valued matrix, i.e., $h_{ij}\in \mathbb{E}$ is a fuzzy number and ${\bf x},{\bf x}_0 \in \mathbb{E}^N$.
Let $[{\bf H}]^{\alpha}=[\underline {\bf H}_{\alpha}, \bar {\bf H}_{\alpha}]$ for any $\alpha \in [0,1]$ be an interval matrix, where
$\underline {\bf H}_\alpha, \bar {\bf H}_\alpha$ denote the matrices whose elements are, respectively, the lower and the upper end points of the interval.

The linear and stationary C-DFS (\ref{linfuzzysystemfull}) can be rewritten  as the FDI:
\begin{equation}
\label{linincl_fuzzy1}
{\bf x}^{\alpha}(k+1)\in [{\bf H}]^{\alpha}[{\bf x}(k)]^{\alpha}; \quad {\bf x}^{\alpha}(0)\in [{\bf x}_0]^\alpha; \quad  0\leq \alpha \leq 1
\end{equation}

It is possible to further specify the results of Theorem \ref{stabfuz}.

\begin{corollary}
\label{zerozero}
Let a Linear and Stationary FDI (\ref{linincl_fuzzy1}); then:
\begin{enumerate}
\item if $\forall U \in [{\bf H}]^0$ is stable then the fuzzy system (\ref{linincl_fuzzy1}) is stable.
\item if $\forall U \in [{\bf H}]^0$ is asymptotically stable then the fuzzy system (\ref{linincl_fuzzy1}) is stable.
\end{enumerate}
\end{corollary}
\begin{proof}
If the first condition is verified, the attainable set $[\mathbb{A}^f({\bf x}_0,k)]^0$ is bounded and the inclusion 
(\ref{aaaa_di}) is stable.
From Theorem \ref{stabfuz}, we have that also the linear and stationary FDI (\ref{linincl_fuzzy1}) is stable.
For asymptotical stability note that if the second condition is verified, then the fuzzy System (\ref{linincl_fuzzy1}) is asymptotically stable, again by Theorem \ref{stabfuz}.
\end{proof}

The above corollary, therefore, requires the inspection of all the dynamic matrices in the interval $[{\bf H}]^0$ to prove the stability; under rather general hypotheses, the following result provides more operative necessary conditions.

\begin{corollary}
\label{teolin1}
Let a Linear and Stationary FDI  (\ref{linincl_fuzzy1}) and suppose that the crisp matrix
$\underline {\bf H}_0 \geq 0$ and that, for all $i=1, \ldots, N$ 
 \begin{equation}
 \label{gersh}
 \sum_{j=1; j\neq i}^N (\bar {\bf H}_0)_{ij}< 1 - (\bar {\bf H}_0)_{ii}
 \end{equation}
then the FDI (\ref{linincl_fuzzy1}) is asymptotically stable.
Analogously the result holds if $\bar {\bf H}_0 \leq 0$ and
 \begin{equation}
 \label{gersh2}
 \sum_{j=1; j\neq i}^N  (\underline {\bf H}_0)_{ij}> -1 - (\underline {\bf H}_0)_{ii}
 \end{equation}
 \end{corollary}

\begin{proof}
Since $\underline {\bf H}_0 \geq 0$ it follows that for each matrix $U\in [H]_0$, $U \geq 0$; moreover $\bar {\bf H}_0\geq0$ is asymptotically stable, since by  Gershgorin circle Theorem \cite{CIRCLE}, its eigenvalues lie, in the complex plane, in the  union of circles centered in $\bar h^0_{ii}\in [0,1]$ with radius equal to $r^{i}_{\bar {\bf H}_0}=\sum_{j=1; j\neq i}^N  (\bar {\bf H}_0)_{ij}\leq 1 - (\bar {\bf H}_0)_{ii}$.

 Since $\bar {\bf H}_0\geq U$, for all $U\in [H]^0$, we have that 
 
 \begin{equation}
 \max_i (u_{ii}+r^i_U)\leq \max_i ((\bar {\bf H}_0)_{ii}+r^i_{\bar {\bf H}_0})< 1
 \end{equation} \noindent
  and  each matrix $U$ is asymptotically stable in the discrete-time sense. 
  According to Corollary \ref{zerozero} the fuzzy System (\ref{linincl_fuzzy1}) is stable.
  The proof for the negative case is analogous.
\end{proof}
Notice that the condition $\underline {\bf H}_0 \geq 0$  includes several systems of practical interest as, for example, the class of positive system whose Gershgorin circles are contained inside the unit ball. These systems find large applications in vague data applications (see, for instance, \cite{possystem1,possystem2}).

The following corollary applies to a generic linear and stationary FDI.
\begin{corollary}
\label{corfuck}
Let a Linear and Stationary FDI  (\ref{linincl_fuzzy1}), and consider the interval matrix $[{\bf H}]^0=[\underline {\bf H}_0, \bar {\bf H}_0]$.

Define ${\bf H}_{0c}=\frac{1}{2}(\underline {\bf H}_0+ \bar {\bf H}_0)$ and $\Delta {\bf H}=\frac{1}{2}(\bar {\bf H}_0- \underline {\bf H}_0)$.
Let $[r]^0=[\underline r_0, \bar r_0]$ and $[i]^0=[\underline i_0, \bar i_0]$ be defined as follows:

\begin{eqnarray}
\underline r_0= \min_{||x||_2=1}\{x^T{\bf H}_{0c}x-|x|^T\Delta {\bf H} |x| \}\\
\bar r_0= \max_{||x||_2=1}\{x^T{\bf H}_{0c}x-|x|^T\Delta {\bf H} |x| \}\\
\underline i_0= \min_{||[x_1,x_2]||_2=1}\{x_1^T({\bf H}_{0c}-{\bf H}_{0c}^T)x_2-\Delta {\bf H} \circ |x_1x_2^T- x_2x_1^T| \}\\
\bar i_0= \min_{||[x_1,x_2]||_2=1}\{x_1^T({\bf H}_{0c}+{\bf H}_{0c}^T)x_2-\Delta {\bf H} \circ |x_1x_2^T- x_2x_1^T| \}
\end{eqnarray}
where $A \circ B= \sum_i \sum_j a_{ij}b_{ij}$ is the scalar product of the matrices $A$ and $B$ and $[x_1,x_2]$ is the stack vector composed of $x_1$ and $x_2$ (where the transposition of vectors has been omitted for brevity).

If the following condition holds true
\begin{equation}
\label{condeig}
\max\{||\underline r_0+j \underline i_0||, ||\underline r_0+j \bar i_0||, ||\bar r_0+j \underline i_0||, ||\bar r_0+j \bar i_0||\}< 1
\end{equation}
where $j$ is the imaginary unit, then the FDI (\ref{linincl_fuzzy1}) is asymptotically stable.

\end{corollary}
\begin{proof}
In \cite{fuzeig} it is proven that, given a square interval matrix $[\underline {\bf H}_0, \bar {\bf H}_0]$, for each matrix ${\bf H}_{0i}\in [\underline {\bf H}_0, \bar {\bf H}_0]$ the eigenvalues $\lambda_{0ij}$ of ${\bf H}_{0i}$ are such that:
\begin{eqnarray}
\underline r_0 \leq Re[\lambda_{0ij}] \leq \bar r_0\\
\underline i_0 \leq Im[\lambda_{0ij}] \leq \bar i_0
\end{eqnarray}

Therefore, if condition (\ref{condeig}) holds true each matrix ${\bf H}_{0i}$ is asymptotically stable in the discrete-time sense. Hence, according to Corollary \ref{zerozero} the fuzzy System (\ref{linincl_fuzzy1}) is stable. 
\end{proof}

Notice that the above operative conditions are just sufficient conditions, and that such conditions do not address the case where the system is stable but not asymptotically stable, because these methods can not determine the multiplicity of the eigenvalues on the unit circle of the matrices in the interval $[H]^0$.

The following corollary provides a simple sufficient condition to determine the simple stability of the fuzzy system.

\begin{corollary}
Let a Linear and Stationary FDI  (\ref{linincl_fuzzy1}) and suppose that there exists a transform matrix $T$ such that

\begin{equation}
\label{gersh3}
T^{-1}\bar {\bf H}_0 T =\begin{bmatrix}
\bar {\bf H}^*_0 & \star \\
0 &1
\end{bmatrix} \quad \mbox{or} \quad T^{-1}\bar {\bf H}_0 T =\begin{bmatrix}
\bar {\bf H}^*_0 & 0 \\
\star &1
\end{bmatrix}
\end{equation}

\begin{itemize}
\item If the conditions of Corollary \ref{teolin1} hold, but inequality (\ref{gersh}) is not strictly satisfied (i.e., $ \sum_{j=1; j\neq i}^N (\bar {\bf H}_0)_{ij}\leq 1 - (\bar {\bf H}_0)_{ii}$) and  inequality (\ref{gersh}) holds strictly for the $N-1 \times N-1$ matrix $\bar {\bf H}^*_0$, then the FDI (\ref{linincl_fuzzy1}) is stable.
\item If the conditions of Corollary \ref{teolin1} hold, but inequality (\ref{gersh2}) is not strictly satisfied and condition (\ref{gersh3}) holds for $\underline {\bf H}_0$ (inequality (\ref{gersh2}) holds strictly for matrix $\underline {\bf H}^*_0$), then the FDI (\ref{linincl_fuzzy1}) is stable.
\item If the conditions of Corollary \ref{corfuck} hold for the interval matrix $[\underline {\bf H}^*_0, \bar {\bf H}^*_0]$ and condition (\ref{gersh3}) holds for both $\underline {\bf H}_0$ and $\bar {\bf H}_0$, then the FDI (\ref{linincl_fuzzy1}) is stable.
\end{itemize}
\end{corollary}
\begin{proof}
From Gershgorin circle Theorem \cite{CIRCLE} the eigenvalues of each matrix $U\in [{\bf H}]^0$ are such that $||\lambda_i||\leq 1$ and $Re[\lambda_i]\geq 0$.

In order to prove that the system is stable (but not asymptotically stable), we need to prove that the multiplicity of eigenvalues with $||\lambda_i||=1$ is exactly $1$. 

Clearly, under the first set of hypotheses matrix Corollary \ref{teolin1} holds strictly for the reduced matrix $\bar {\bf H}^*_0$, and therefore it has eigenvalues within the open unitary circle in the complex plane; due to the triangularity of the transformed matrix, the statement is verified. An analogous argument proves the statement for the other cases.\end{proof}

The following section describes the practical representation of a linear and stationary Discrete-time Fuzzy System, and is the discrete-time analogous of the results in \cite{Laksh:2003, Diamond}.

\section{Representation of  Linear Fuzzy Difference Inclusions}
\label{sec_fuzzy2}
In order to study the practical representation of a linear and stationary Discrete-time Fuzzy System
there is the need to discuss the solution set of System (\ref{linincl_fuzzy1}).
Consider a linear crisp system ${\bf z}(k+1)={\bf G}(k){\bf z}(k), {\bf z}(0)={\bf z}_0$. The solution of such system is given by
\begin{equation}
\label{eq:vc_crisp_sol}
{\bf z}(k)=\Phi(k){\bf z}_0
\end{equation}
where $\Phi(k)$ satisfies the matrix difference equation
\begin{equation}
\label{eq:phi_mdi}
\Phi(k+1)={\bf G}(k)\Phi(k); \quad \Phi(0)=I_N
\end{equation}
Clearly, if the system is also stationary, $\Phi(k)=G^k$.

Analogously to the crisp case, consider the fuzzy version of System (\ref{eq:vc_crisp_sol}); then, equation (\ref{eq:phi_mdi}) can be interpreted as the family of difference inclusions
\begin{equation}
\label{eq:vc_fuz}
\Phi_\alpha(k+1)\in [{\bf H}(k)]^\alpha\Phi_\alpha(k); \quad \Phi_\alpha(0)=I_N; \quad 0\leq \alpha \leq 1
\end{equation}

Denote
\begin{equation}
\Phi_	\alpha(k)=\{Y(k): Y(k+1)=U(k)Y(k), Y(0)=I, U(\cdot)\in [H]^\alpha(\cdot)\}
\end{equation}

From the Stacking Theorem \ref{stacking}, the $\Phi_\alpha(\cdot)$ are the level sets of the $\mathbb{D}^{N\times N}$ valued fuzzy function $\Phi(\cdot)$. If the matrix $H$ is stationary and $[H]^\alpha$ is an interval matrix for each $\alpha \in [0,1]$, then $\Phi(\cdot)$ is $\mathbb{E}^{N\times N}$ valued, since an interval matrix is a convex set in $\mathbb{R}^{N\times N}$.
In this case $U(k)=U$ and $Y(k)=U^k$, hence
\begin{equation}
\Phi_	\alpha(k)=\{U^k: U\in [H]^\alpha\}=[H^k]^\alpha
\end{equation}
In this case denote the fuzzy set $\Phi(k)$ as $H^k$.

In the following theorem it is proved that, if $\underline H_\alpha$ is nonnegative, then, for each level set, the evaluation of bounds of the set of solutions is simplified.

Notice that, within the proposed approach, the solution is not a single trajectory, but it is defined as a set; there is the need to consider, therefore, the attainable set and the set of solutions.
 
In order to derive a framework with a real applicability, the following theorem proves that, limiting the scope to the class of linear and stationary FDI with non-negative entries, the evaluation of the upper and lower bounds of the set of the solutions is extremely simplified.

\begin{theorem}
\label{pippo1}
Let a linear and stationary FDI (\ref{linincl_fuzzy1}) and suppose that the crisp matrix $\underline H^0 \geq 0$; then the following holds for each $\alpha \in [0,1]$:

\begin{equation}
\underline{\mathbb{S}}_\alpha({\bf x}_0,k)=(\underline{H}_\alpha)^k \underline{{\bf x}}_{\alpha0}; \quad \bar{\mathbb{S}}_\alpha({\bf x}_0,k)=(\bar{H}_\alpha)^k \bar{{\bf x}}_{\alpha0}
\end{equation}

\noindent where  $\underline{{\bf x}}_{\alpha0}$ and $\bar{{\bf x}}_{\alpha0}$ are the left and right bounds of $[{\bf x}_0]^\alpha$, respectively and $\underline{\mathbb{S}}_\alpha({\bf x}_0,k)$, $\bar{\mathbb{S}}_\alpha({\bf x}_0,k)$ are the left and right bounds of  the set of solutions $\mathbb{S}_\alpha({\bf x}_0,k)$ of the single Inclusion in the family (\ref{incl_fuzzy1}).
\end{theorem}

\begin{proof}
We have to show that, given two matrices $R$ and $G$, with $G\geq R\geq0$ such that
\begin{eqnarray*}
X(k+1)=RX(k), \quad X(0)=I\\
Y(k+1)=GY(k), \quad Y(0)=I 
\end{eqnarray*}
it follows that $Y(k)\geq X(k), k\geq 0$. Clearly, $X(k)=R^k$ and $Y(k)=G^k$. Since $G\geq R\geq0$ the statement is proved.
Note that if $\underline {\bf H}_0 \geq 0$ then the level set matrices $\underline H^k_\alpha$ and $\bar H^k_\alpha$ are positive matrices for each $\alpha\in[0,1]$, and when multiplying interval vectors the order of endpoints is preserved.
We have therefore that $\mathbb{S}_\alpha({\bf x}_0,k)=[\underline{\mathbb{S}}_\alpha({\bf x}_0,k), \bar{\mathbb{S}}_\alpha({\bf x}_0,k)]$
and the theorem is proved.
\end{proof}

The theorem proves that, under the hypotheses, for each $\alpha$-level the upper and lower bound of the set of solutions can be calculated independently as if they were two crisp systems with dynamic matrices $\underline{H}_\alpha$ and $\bar{H}_\alpha$, respectively; therefore, for practical applications, it is sufficient to consider the bounds to adequately represent the solution set.
\section{Conclusions}
\label{sec_conclusions}

In this paper the stability and representation of linear and stationary fuzzy systems where both the dynamic coefficients and the state variables are described by means of fuzzy numbers is addressed, providing some stability conditions.

Further works will be devoted to find stability conditions for time varying and nonlinear fuzzy systems.

\end{document}